\documentclass[orivec]{llncs}
\usepackage{url}

\usepackage{amssymb}
\usepackage{latexsym}
\usepackage{amsmath}
\usepackage{amsthm}
\usepackage{xcolor}
\usepackage{mathtools}

\usepackage{comment}

\usepackage{algpseudocode}
\usepackage{algorithm}

\DeclareMathOperator*{\argmin}{argmin}

\newtheorem*{ESS}{Equal Sum Subsets problem (ESS)}
\newtheorem*{FSSR}{Family of Subset-Sum Ratio problems (F-SSR)}
\newtheorem*{SRFSSR}{Family of Semi - Restricted Subset-Sum Ratio Problems \\ (\textsl{Semi-Restricted F-SSR})}

\newtheorem*{TSSSR}{Two-Set Subset-Sum Ratio problem (2-Set SSR)}
\newtheorem*{facSSR}{Factor-$r$ Subset-Sum Ratio problem (Factor-$r$ SSR)}

\newtheorem*{SSRmatroids}{Subset-Sum Ratio with Matroid constraints}

\newtheorem*{SSR}{Subset-Sums Ratio problem (SSR)}

\begin{document}
\mainmatter              
\title{Approximation Schemes for Subset Sum Ratio Problems}
\titlerunning{FPTAS for SSR Variations}  
%
\author{Nikolaos Melissinos\inst{1}, Aris Pagourtzis\inst{2} \and Theofilos Triommatis\inst{3}}
\authorrunning{Melissinos et al.} 
%
\tocauthor{Nikolaos Melissinos, Aris Pagourtzis, Theofilos Triommatis}

\institute{
Universit\'e Paris-Dauphine, PSL University, CNRS, LAMSADE, \\75016 Paris, France,\\
\email{{nikolaos.melissinos@dauphine.eu}}
\and
School of Electrical and Computer Engineering\\
National Technical University of Athens\\
Polytechnioupoli, 15780 Zografou, Athens, Greece,\\
\email{pagour@cs.ntua.gr}
\and
School of Electrical Engineering, Electronics and Computer Science\\ University of Liverpool, 
Liverpool, L69-3BX, UK \\
\email{Theofilos.Triommatis@liverpool.ac.uk}
}

\maketitle              

\begin{abstract}
We consider the Subset Sum Ratio Problem ($SSR$), in which given a set of integers the goal is to find two subsets such that the ratio of their sums is as close to~1 as possible, and introduce a family of variations that capture additional meaningful requirements. Our main contribution is a generic framework that yields fully polynomial time approximation schemes (FPTAS) for problems in this family that meet certain conditions. We use our framework to design explicit FPTASs for two such problems, namely \textsl{Two-Set Subset-Sum Ratio} and \textsl{Factor-$r$ Subset-Sum Ratio}, with  running time $\mathcal{O}(n^4/\varepsilon)$, which coincides with the best known running time for the original $SSR$ problem~\cite{MelissinosP18}. 

\end{abstract}

\begin{section}{Introduction}
\label{secIntro}
Subset sum computations are of key importance in computing, as they appear either as standalone tasks or as subproblems in a vast amount of theoretical and practical methods coping with important computational challenges. As most of subset sum problems are NP-hard, an effort was made over the years to come up with  systematic ways of deriving approximation schemes for such problems.
Important contributions in this direction include works by Horowitz and Sahni~\cite{HorowitzS74}, \cite{HorowitzS76}, Ibarra and Kim~\cite{IbarraK75}, Sahni~\cite{Sahni76}, Woeginger~\cite{Woeginger00} and Woeginger and Pruhs~\cite{PruhsW07}. Inspired by these works we define and study families of variations of the Subset Sum Ratio Problem 
which is a combinatorial optimization problem  introduced and shown NP-hard by Woeginger and Yu~\cite{wo:yu}. The formal definition of the problem is as follows:  

\begin{SSR}
Given a set $A = \{a_1,\ldots,a_n\}$ of $n$ positive integers, find two 
nonempty and disjoint sets $ S_1, S_2 \subseteq \{1,\ldots,n\}$ that minimize the ratio 
\begin{align*}
 &\frac{\max\{\sum_{i \in S_1} a_i, \sum_{j \in S_2} a_j\}}{\min\{\sum_{i \in S_1} a_i, \sum_{j \in S_2} a_j\}} 
\end{align*}
\end{SSR}



One of our motives to study $SSR$ stems from the fact that it is the optimization version of the decision problem Equal Subset Sum (ESS) which is related to various other concepts and problems as we explain below. ESS is defined as follows: 

\begin{ESS}
Given a set $A = \{a_1,\ldots,a_n\}$ of $n$ positive integers, are there two 
nonempty and disjoint sets $ S_1$, $S_2$ $\subseteq \{1,\ldots,n\}$ such that 
$ \sum_{i \in S_1} a_i = \sum_{j \in S_2} a_j \mbox{?} $
\end{ESS}

Even if this problem has been in the literature for many years, it is still being studied with a recent work begin \cite{mu:ne:pa:we}.
Variations of this problem have been studied and proven NP-hard by Cieliebak
\textsl{et al.}\ in~\cite{cie:eid:pag,cie:eid:pag:sch}, where  pseudo - polynomial time algorithms were also presented for many of these problems.
ESS is a
fundamental problem appearing in many scientific areas. For example, it is related to the Partial Digest problem that comes from molecular biology~\cite{cie:eid:pen,cie:eid}, 
to allocation mechanisms~\cite{lip:mar:mos:sab}, to tournament construction~\cite{khan}, to a variation of the 
Subset Sum problem, namely the Multiple 
Integrated Sets SSP, which finds applications in cryptography~\cite{vol}.
A restricted
version of $ESS$, namely when the sum of the input values is strictly less than
$2^n-1$ is guaranteed to have a solution, however it is not known how to find it;
this version belongs to the complexity class PPP \cite{pap94}.

The first FPTAS for $SSR$ was introduced by Bazgan \textsl{et al.} in \cite{baz:san:tuz} and more recently a simpler but slower FPTAS was introduced in~\cite{Nanon} and a faster one in~\cite{MelissinosP18}; the latter is the fastest known so far for the problem. Variations of ESS were studied and shown NP-hard in~\cite{cie:eid:pag:sch:tech,cie:eid:pag,cie:eid:pag:sch}, where also pseudo - polynomial time algorithms were presented for some of them and it was left open whether the corresponding optimization problems admit an FPTAS. Here we address that question in the affirmative for two of those problems (namely for \textsl{Equal Sum Subsets From Two Sets}~\cite{cie:eid:pag:sch:tech}\footnote{It is not hard to show that the optimization version of \textsl{Equal Sum Subsets From Two Sets} can be reduced to \textsl{Two-Set Subset-Sum Ratio} for which we provide an FPTAS here.} and for \textsl{Factor-$r$ Sum Subsets}~\cite{cie:eid:pag:sch}) and provide a framework that can be potentially  used to give an FPTAS for most
of the remaining ones, if not for all, as well as for many other subset sum ratio problems. 

Let us note that, for the exemplary problems that we study here there may exist more efficient approximation algorithms, e.g.\ by using techniques such as those in \cite{jin19,chan18} for knapsack, however it is not clear if and how such techniques can be adapted in a generic way to take into account the additional restrictions that are captured by our framework. Moreover, our primary goal is to provide an as generic as possible framework to cope with such problems, at the cost of sacrificing optimality in  efficiency.

Our results and organization of the paper are as follows. 
In section~\ref{families} we define two families of variations of $SSR$ problems that are able to capture additional restrictions. Our main result, presented in section~\ref{SecExFPTAS}, is a method to obtain an FPTAS for any problem in these families the definition of which meets certain conditions. In the last two sections we use our framework to present FPTASs  for two variations of $SSR$, namely \textsl{2-Set SSR} and \textsl{Factor-$r$ Subset Sum Ratio}; to the best of our knowledge, no approximation algorithm was known so far for these problems.

\end{section}

\begin{section}{Families of Variations of SSR}
\label{families}
In this section we shall define two families of variations of the $SSR$ problem. In~\cite{MelissinosP18}, the function $\mathcal{R}(S_1,S_2,A)$ was defined as: 
\begin{definition}[Ratio of two subsets]
Given a set $A = \{a_1,\ldots,a_n\}$ of $n$ positive integers and two 
sets $ S_1, S_2 \subseteq \{1,\ldots,n\}$ we define 
$\mathcal{R}(S_1,S_2,A)$ as follows:
\begin{align*}
 \mathcal{R}(S_1,S_2,A)= 
 \begin{cases}
  0 & \mbox{ if } S_1 = \emptyset \mbox{ and } S_2 \neq \emptyset \\
  \frac{\sum_{i \in S_1} a_i}{\sum_{i \in S_2} a_i} &\mbox{ if } S_2 \neq \emptyset,\\
  +\infty &\mbox{ otherwise.}
 \end{cases}
\end{align*}
\end{definition}

Here we will also define and use  $\mathcal{MR}(S_1,\ldots,S_k,A)$ which is a generalization of $\mathcal{R}(S_1,S_2,A)$ to $k>2$ sets:
\begin{definition}[Max ratio of $k$ subsets]
Given a set $A = \{a_1,\ldots,a_n\}$ of $n$ positive integers and two 
sets $ S_1,\ldots, S_k \subseteq \{1,\ldots,n\}$ we define:
\begin{align*}
 \mathcal{MR}(S_1,\ldots,S_k,A)= \max\{ \mathcal{R}(S_i,S_j,A)\mid i\neq j  \mbox{ and } i,j \in \{1,\dots,k\}\}\enspace
\end{align*}
\end{definition}

In order to keep our expressions as simple as possible we will use the above functions throughout the whole paper.  

Let us now define the first family of variations of $SSR$. We want this family to contain as many problems as possible. In a general case we may not have just a set of numbers as our input but a graph or that has a weights on the edges or the vertices. For such reasons we will use the following notation.

\begin{FSSR}  \label{fssr}
A problem $\mathcal{P}$ in $F$-$SSR$ is a combinatorial optimization problem $(\mathcal{I}, k, \mathcal{F})$ where:
\begin{itemize}
    \item $\mathcal{I}$ is a set of instances each of which is a pair $(E,w)$ where $E=\{e_1,\ldots ,e_n\}$ is a set of ground elements and $w: E \mapsto \mathbb{R}^+ $ is a weight function which maps every element $e_i$ to a positive number $a_i$;
    \item $k$ defines the number of subsets of $\{1,\ldots,n\}$ we are searching for;
    \item $\mathcal{F}$ gives the set of feasible solutions as follows: for any input $(E,w)$, $\mathcal{F}(k,E)$ is a collection of $k$-tuples of nonempty and disjoint subsets of $\{1,\ldots,$ $n\}$, and
    given $(k,E, (S_1,\ldots,S_k))$ we can check in polynomial time whether $(S_1,\ldots,$ $S_k)$ $\in \mathcal{F}(k,E)$.
\end{itemize}

\noindent
The goal of $\mathcal{P}$ is to find for an instance $(E,w)$ a feasible solution $(S^*_1,\ldots,S^*_k)$ such that
\begin{align*}
\mathcal{MR}(S^*_1,\ldots,S^*_k, A )= \min\{\mathcal{MR}(S_1,\ldots,S_k, A) \mid (S_1,\ldots,S_k) \in \mathcal{F}(k,E) \} 
\end{align*}
where $A=\{a_i = w(e_i) \mid e_i \in E \}$
\end{FSSR}

\begin{remark}
\label {remIndex}
Note that under this definition of $F$-$SSR$ the function $w$ of an instance $(E,w)$ does not play any role in deciding whether a $k$-tuple $(S_1,\ldots,S_k)$ is feasible or infeasible solution; in other words, the element weights do not affect feasibility, only their indices do. Consequently, for a specific problem $\mathcal{P} =(\mathcal{I},k,\mathcal{F}) \in F\mbox{-}SSR$ and two different instances $(E,w)$ and $(E,w')$ in $\mathcal{I}$ with the same ground elements $E$, the feasible solutions of the two instances are the same.
\end{remark}

We will now introduce a family that is similar to $F$-$SSR$ 
but there is a major difference which is an extra condition. In this family we know (we give it as input), the minimum between the maximum values of the solution. This is rather technical and it will become obvious in the following paragraphs.

\begin{SRFSSR}\label{srfssr}
For every problem $\mathcal{P}=(\mathcal{I},k, \mathcal{F})$ in $F$-$SSR$, we define an associated optimization problem $\mathcal{P}' = (\mathcal{I}', k', \mathcal{F}')$ as follows:
\begin{itemize}
    \item the set of instances of $\mathcal{P'}$ is  \\ \centerline{$ \mathcal{I}'=\{(E,w, m) \mid (E,w)\in \mathcal{I} \mbox{ and } m\in \{1,\ldots |E| \}  \} $}
    \item $k'=k$
    \item the collection  of feasible solutions of instance $(E,w, m) \in \mathcal{I}'$ is given by: $$\mathcal{F}'(k,E,w, m) = \{ (S_1,...,S_k) \in \mathcal{F}(k,E) \mid \min_{j \in \{1,\ldots, k\}}\{ \max_{i\in S_j }w(e_i)\} = w(e_m)\}  $$
\end{itemize}

\noindent
and the goal of $\mathcal{P}'$ is to find for an instance $(E,w,m)$ a feasible solution \hfill \\ $(S^*_1,$ $,\ldots,$ $S^*_k)$ such that
\begin{align*}
\mathcal{MR}(S^*_1,\ldots,S^*_k, A )= \min\{\mathcal{MR}(S_1,\ldots,S_k, A) \mid (S_1,\ldots,S_k) \in \mathcal{F}(k,E) \} 
\end{align*}
where $A=\{a_i = w(e_i) \mid e_i \in E \}$.
We define the family of problems $Semi$- $Restricted$ $F$-$SSR$ as the class of problems $\left\{\mathcal{P}' \mid \mathcal{P} \in F \mbox{-}SSR\right\}$.
\end{SRFSSR}

\begin{remark}
We note that if a problem belongs to $Semi\mbox{-}Restricted$ $F\mbox{-}SSR$ it cannot belong to $F\mbox{-}SSR$ because there is the extra condition for a solution $(S_1,...,S_k)$ to be feasible, $\min_{j \in \{1,\ldots k\}}\{ {\max_{i\in S_j }w(e_i)}\} = w(e_m)$  which depends on the weight function $w$ and not only on the set of elements $E$ as is the case for problems in $F\mbox{-}SSR$.
\end{remark}
 
\begin{remark}
It is obvious that if there exists a deterministic polynomial time Turing Machine that can decide if a solution is feasible for a problem $\mathcal{P} $ in F-SRR then we can construct another deterministic polynomial time Turing Machine that takes into account the extra condition to decide, if a solution is feasible for the semi restricted version  $\mathcal{P}' $ in Semi Restricted F-SRR of the previous. 
\end{remark}

We must note that $F\mbox{-}SSR$  contains many problems of many from different areas in computer science and could prove useful to get an FPTAS for them if we could develop a pseudo - polynomial algorithm with a particular property (it will be explained in the next section) for the semi restricted versions of them. This family includes subset sum ratio problems with matroid restrictions, graph restrictions, cardinality restrictions (including partition problems). Some more problems could be scheduling problems and knapsack.

To give some examples, we will present some problems that belong in $F$-$SSR$. For the first two, the proof that they actually do belong in $F$-$SSR$ will be presented in section~\ref{secTS} and  section~\ref{SecExamples} respectively. We must note that  the decision version of $Factor$-$r$ $SSR$ that follows was studied in \cite{cie:eid:pag:sch}. For these two problems, we will introduce FPTAS algorithms in the following sections. Moreover we will present other problems of $F\mbox{-}SSR$ that have more complicated constraints and could prove interesting to be studied in the future.

\begin{TSSSR}
Let $A=\{(a_1,b_1),$ $\dots,(a_n,b_n)\}$ be a set of pairs of real numbers. We are searching for two nonempty and disjoint sets $S_1,S_2 \subseteq \{1,...,n\}$ that minimize
\begin{align*}
 &\frac{\max\{ \sum_{i \in S_1} a_i, \sum_{j \in S_2} b_j\}}{\min\{ \sum_{i \in S_1} a_i, \sum_{j \in S_2} b_j\}} \ .
\end{align*}
\end{TSSSR}

\begin{facSSR}
Given a set $A = \{a_1,\ldots,a_n\}$ of $n$ positive integers and a real number $r\geq1$, find two 
nonempty and disjoint sets $ S_1$, $S_2$ $\subseteq \{1,\ldots,n\}$ that minimize the ratio 
\begin{align*}
 &\frac{\max\{r \cdot \sum_{i \in S_1} a_i, \sum_{j \in S_2} a_j\}}{\min\{r \cdot \sum_{i \in S_1} a_i, \sum_{j \in S_2} a_j\}} \ .
\end{align*}
\end{facSSR}

In \cite{go:mo:tl} there were introduced digraph constraints for the subset sum problem which can easily be modeled via our framework. Generally we are able demand as constraints of $S_1$ and $S_2$ to be a specific property considering the vertices of the graph, for example we may demand that the solution consists of independent sets or dominant sets etc. Not only can we impose constraints for the sets of vertices but we can impose constraints on the edges of the graph as well. Finally we may impose restrictions that concern both edges and vertices at the same time, for example take into account vertices that form a complete graph. 

The same way we define the constraints from a graph we are able to demand that the solution of a problem  consists of independent sets of a given matroid. 

\begin{SSRmatroids}
Given a matroid $\mathcal{M}(E,I)$ and a weight function $w : E \rightarrow \mathbb{R}^+$. We want to find two non empty and non equal sets $S_1,S_2 \in  I $ such that:
\begin{align*}
 &\frac{\max\{\sum_{x \in S_1} w(x), \sum_{y \in S_2} w(y)\}}{\min\{\sum_{x \in S_1} w(x), \sum_{y \in S_2} w(y)\}} \ .
\end{align*}
\end{SSRmatroids}

\noindent
Before we continue to the next section we will present two lemmas which give us information about the  solutions which are feasible for both problems in $F\mbox{-}SSR$ and $Semi$-$Restricted$ $F\mbox{-}SSR$. Moreover we must note that
all the proofs for the theorems and the lemmas can be found in the appendix.
\begin{lemma}
\label{lemma1}
Let $\mathcal{P}=(\mathcal{I},k,\mathcal{F}) $ a problem in $F\mbox{-}SSR$ and  $\mathcal{P'}=(\mathcal{I}',k',\mathcal{F}') $ the semi restricted version of $\mathcal{P}$. If $(E,w) \in \mathcal{I}$ and $(E,w',m) \in \mathcal{I}'$ are the instances of $\mathcal{P}$ and $\mathcal{P}'$ respectively then 
    any feasible solution $(S_1,\ldots, S_k)$  of the instance $(E, w',m)$ of $\mathcal{P}'$
    is also a feasible solution of the instance $(E,w)$  of $\mathcal{P}$.
\end{lemma}

\begin{lemma} \label{lemma2}
Let $\mathcal{P}=(\mathcal{I},k,\mathcal{F}) $ a problem in $F\mbox{-}SSR$ and  $\mathcal{P'}=(\mathcal{I}',k',\mathcal{F}') $ the semi restricted version of $\mathcal{P}$. If $E$ is a set of elements and $w$, $w'$ two weight functions such that:
$$\mbox{For any } i,j \in \{1,\ldots n\}, \ \  w(e_i)< w(e_j) \Leftrightarrow w'(e_i)\ \leq w'(e_j) $$
then any feasible solution $(S_1,\ldots, S_k)$ for the instance $(E,w)$  of $ \mathcal{P}$ is a feasible solution for the instance $(E,w',m)$  of $ \mathcal{P}'$ if 
$$w(e_m) = \min_{j \in \{1,\ldots,k\}}\{ \max\{ w(e_i) \mid i \in S_j\}\}$$
\end{lemma}

\end{section}

\begin{section}{A Framework Yielding FPTAS for Problems in $F$-$SSR$}
\label{SecExFPTAS}

In the following theorem we want to define a scale parameter $\delta $ which we will use later to change the size of our input and the pseudo - polynomial algorithms  will run in polynomial time. In addition by using these parameters, we will define the properties that the sets of the output should satisfy to be $(1+\varepsilon)$ approximation. These parameters are not unique but any other number should to the trick as long as all the properties of the theorem bellow are satisfied. 

\begin{theorem}
\label{theor10}
Let $A = \{a_1,...,a_n\}$ be a set of positive real numbers, $\varepsilon \in (0,1)$, two sets $S_{1Opt}, S_{2Opt} \subseteq \{1,...,n\}$ and any numbers $w,m,\delta$ that satisfy: 
\begin{itemize}
    \item $0< w \leq \min{\left(\sum_{i \in S_{1Opt}}{a_i},\sum_{i \in S_{2Opt}}{a_i}\right)}$ 
    \item $n \geq m \geq \max{\left(|S_{1Opt}|,|S_{2Opt}|\right)},$
    \item $\delta = (\varepsilon \cdot w)/(3 \cdot m)$
\end{itemize}
If $S_1, S_2 \subseteq \{1,...,n\}$ are two non-empty sets such that:
\begin{itemize}
    \item $w \leq \min{\left(\sum_{i \in S_{1}}{a_i},\sum_{i \in S_{2}}{a_i}\right)}$ \hfill
    \item $n \geq m \geq \max{\left(|S_{1}|,|S_{2}|\right)},$
    \item  $1 \leq \mathcal{MR}(S_1,S_2,A') \leq  \mathcal{MR}(S_{1Opt},S_{2Opt},A')$ where $A' = \{ \lfloor \frac{a_1}{\delta} \rfloor,...,\lfloor\frac{a_n}{\delta}\rfloor\}$
\end{itemize}
Then the following inequality holds
\begin{align*}
    1 \leq \mathcal{MR}(S_1,S_2,A) \leq  (1+\varepsilon) \cdot \mathcal{MR}(S_{1Opt},S_{2Opt},A) \ .
\end{align*}
\end{theorem}

The next theorem presents the conditions that should be met to construct an FPTAS algorithm for a problem that belongs in $F$-$SSR$. 
Keep in mind that this framework should be considered similar to linear programming, i.e. if there is a way to prove that a problem belongs to $F\mbox{-}SSR$ and at the same time there is a pseudo - polynomial time algorithm for its semi-restricted version then we can obtain an FPTAS algorithm.

\begin{theorem}
\label{theor20}
Let $\mathcal{P} = (\mathcal{I},\mathcal{F},\mathcal{M},\mathcal{G}) $ be a problem in $F$-$SSR$ and   $\mathcal{P}'  = (\mathcal{I}',\mathcal{F}',$ $\mathcal{M},$ $\mathcal{G})$ its corresponding problem in $Semi \ Restricted \ F\mbox{-}SSR$.
If for problem $\mathcal{P}'$ there exists an algorithm that solves exactly all instances $A = \{a_1,\ldots a_n,m\} \in \mathcal{I}'$ in which all $a_i$ values are integers in time $\mathcal{O}(poly(n,a_m))$, then $\mathcal{P}$ admits an FPTAS.
\end{theorem}

Now we will present an algorithm that approximates $\mathcal{P}$ using the algorithm for $\mathcal{P}'$. We will denote the algorithm that returns the exact solution for $\mathcal{P}'$ by $\mathcal{SOL}_{ex,\mathcal{P}'}(A)$.

\begin{algorithm}[H]
\caption{FPTAS for the problem $\mathcal{P}$ [$\mathcal{SOL}_{apx,\mathcal{P}}(A)$ function]}
\label{Alg1}
\begin{algorithmic}[1]
  \Require A set  $A= \{a_1,\ldots,a_n\}, a_i\in \mathbb{R}^+$.
  \Ensure Sets with max ratio $(1+\varepsilon)$ to the optimal max ratio for the problem $\mathcal{P}$.
  \State $(S_1^*,\ldots,S_k^*) \leftarrow \{\emptyset,\ldots,\emptyset\} $
  \For{$m\leftarrow 1$ \textbf{to} $n$}
  \State $\delta \leftarrow \frac{ \varepsilon \cdot a_m}{3 \cdot n}$
  \State $A^{(m)}\leftarrow \emptyset$
  \For{$i\leftarrow 1$ \textbf{to} $n$}
  \State $a'_i \leftarrow \lfloor \frac{a_i}{\delta}\rfloor$
  \State $A^{(m)}\leftarrow A^{(m)} \cup \{a'_i\}$
  \EndFor
  \State ${A}^{(m)} \leftarrow  {A}^{(m)}\cup  \{m\}$
  \State $(S'_1,\ldots, S'_k) \leftarrow \mathcal{SOL}_{ex,\mathcal{P}'}(A^{(m)})$
  \If{$\mathcal{MR}( S'_1,\ldots, S'_k,A) \leq \mathcal{MR}(S_1^*,\ldots, S_k^*,A) $} 
    \State $(S_1^*,\ldots,S_k^*) \leftarrow (S'_1,\ldots, S'_k)$
  \EndIf
  \EndFor
  \State \Return $(S_1^*,\ldots,S_k^*)$
\end{algorithmic}
\end{algorithm}  
\noindent
In the next sections we will give some examples of how this framework works by using Theorem~\ref{theor20} to find an FPTAS algorithm for some problems.

\end{section}

\begin{section}{2-Set SSR}
\label{secTS} Here, we will design an FPTAS algorithm for $2 \mbox{-} Set\ SSR$. We must note that faster approximation algorithms could be developed for this particular problem but this is not the scope of this section.  

We will begin by proving that this problem belongs in $F\mbox{-}SSR$. We will match the $2 \mbox{-}Set \ SSR$ with a problem $(\mathcal{I},\mathcal{F},\mathcal{M},\mathcal{G})$ in $F\mbox{-}SSR$. If we let  the set of instances $\mathcal{I}$ contain sets of positive numbers $A=\{a_1,\ldots,a_{2 \cdot n}\}= \{a_1,\ldots,a_n,$ $b_1,\ldots,b_n\}$, and the set of feasible solutions $\mathcal{F}$ contain all the pairs of sets $(S_1,S_2)$ such that $ S_1 \subseteq \{1,...,n\} \mbox{, } S_2 \subseteq \{n+1,...,2\cdot n\} $, $ \nexists \ (i,j) \mbox{ such that } i \in S_1, j\in S_2 \mbox{ with } i \equiv  j \ (\bmod\ n) $, 
the objective $\mathcal{M}=\mathcal{MR}(S_1, S_2,A))$ and the goal function $\mathcal{G}=\min$,
then the 2-Set SSR problem coincides with $(\mathcal{I,F,M,G})$ which is a problem in $F\mbox{-}SSR$.

Now we shall present a pseudo - polynomial algorithm that finds an optimal solution for $Semi$-$Restricted$ $2\mbox{-}Set$ $SSR$. Our algorithm employs two separate algorithms for two different cases.  

\begin{algorithm}[H]
\caption{$Semi$-$Restricted$ $2 \mbox{-}Set \ SSR$ solution [$\mathcal{SOL}(A,m)$ function]}
\label{AlgTSm}
\begin{algorithmic}[1]
    \Require a set $A=\{a_1,\ldots,a_{2\cdot n}\}$, $a_i\in \mathbb{Z}^+$, and an integer $m$, $1\leq m\leq 2 \cdot n$.
    \Ensure the sets of an optimal solution for $Semi$-$Restricted$ $2 \mbox{-}Set \ SSR$.
    \State $S'_1\leftarrow \emptyset$, $S'_2\leftarrow \emptyset$, $S_{min}\leftarrow \emptyset$, $S_{max}\leftarrow \emptyset$
    \If{$m\leq n$}
        \State $(p,p')\leftarrow(0,n)$
    \ElsIf{$n< m \leq 2\cdot n$}
        \State $(p,p')\leftarrow(n,0)$
    \EndIf
    \State $S_{min}\leftarrow \{i \mid i \in \{1,\ldots ,n\} \mbox{ and }  a_{i+p} \leq a_m \}\smallsetminus\{m - p \}$
    \State $S_{max}\leftarrow \{i \mid i \in \{1,\ldots ,n\} \mbox{ and } a_{i+p'} \geq  a_m \} \smallsetminus\{m -p +p'\} $
    \If{$S_{max} \neq \emptyset$}
        \State $(S_1,S_2)\leftarrow \mathcal{SOL}_{Case1}(A,m,S_{min},S_{max}) $
        \State $(S'_1,S'_2)\leftarrow \mathcal{SOL}_{Case2}(A,m,S_{min},S_{max}) $
        \If{$\mathcal{MR}(S_1,S_2,A) < \mathcal{MR}(S'_1,S'_2,A) $}
            \State $(S'_1,S'_2)\leftarrow (S_1,S_2)$
        \EndIf
    \EndIf
    \State \Return $S'_1$, $S'_2$
\end{algorithmic}
\end{algorithm}  

\noindent
We will continue with the presentation of algorithms $\mathcal{SOL}_{Case1}(A,m,S_{min},$ $S_{max}) $ and $\mathcal{SOL}_{Case2}(A,m,S_{min},S_{max}) $. Let us first define a function that will simplify their presentation.

\begin{definition}[$\mathsf{LTST}$: Larger Total Sum Tuple selection]
Given two tuples $\vec{v_1}=(S_{1},S_{2},x)$ and $\vec{v_2}=(S_{1}',S_{2}',x')$ 
we define the function $\mathsf{LTST}(\vec{v_1},\vec{v_2})$ as follows:
\begin{align*}
 \mathsf{LTST}(\vec{v_1},\vec{v_2})= 
 \begin{cases}
  \vec{v_2} &\mbox{ if } \vec{v_1} = (\emptyset,\emptyset,0) \mbox{ or } x'>x \mbox{,}\\
  \vec{v_1} &\mbox{ otherwise}\enspace .
 \end{cases}
\end{align*}
\end{definition}
\noindent
We will use this function to compare the sum of the sets $S_1 \cup S_2$ and $S_1' \cup S_2'$ i.e. 
\begin{align*}
    x = \sum_{i\in S_1 \cup S_2}{a_i} && \mbox{ and } && x' = \sum_{i\in S_1' \cup S_2'}{a_i}
\end{align*}

In the next algorithm we study the case 1. In case 1 we consider that we need to use an element that its weight is greater than the sum of the elements' weights that could belong to the other set. In this case the set with the largest total weight contains only one element and the other set contains all the allowed elements (elements that have no conflicts).
\begin{algorithm}[H]
\caption{Case 1 solution [$\mathcal{SOL}_{ Case  1}(A,m,S_{min},S_{max})$ function]}
\label{sub2}
\begin{algorithmic}[1]
    \Require a set $A=\{a_1,\ldots,a_{2\cdot n}\}$, $a_i\in \mathbb{Z}^+$ and an integer $m$, $1\leq m\leq 2\cdot n$ and $S_{min},S_{max} \subseteq \{1,...,n\} $.
    \Ensure Case 1 optimal solution for $Semi$-$Restricted$ $2 \mbox{-}Set \ SSR$.
    \State $S'_1\leftarrow \emptyset$, $S'_2\leftarrow \emptyset$
    \If{$m\leq n$}
        \State $p \leftarrow 0$, $p' \leftarrow n$
    \Else
        \State $p \leftarrow n$, $p'\leftarrow 0$
    \EndIf
    \State $Q \leftarrow a_m + \sum_{i \in S_{min} } a_{i+p}$
    \ForAll{$i\in S_{max}$ \textbf{and} $a_{i+p'}>Q$}
        \State $a\leftarrow 0$
        \If{$i \in S_{min}$}
            \State $a\leftarrow a_{i+p}$
        \EndIf
        \If{$ a_{i+p'}/(Q-a) < \mathcal{MR}(S'_1,S'_2,A)$}
            \State $S\leftarrow \{j+p \mid j \in S_{min} \mbox{ or } j=m-p \} \smallsetminus \{i+p \}$
            \State $(S'_1,S'_2)\leftarrow (S, \{i+p'\})$
        \EndIf
    \EndFor
    \State \Return $S'_1$, $S'_2$
\end{algorithmic}
\end{algorithm}  

In case 2 we consider that the largest (weighted) element doesn't necessarily dominate the sum of the weights of the second set. In this case we create a three dimensional matrix whose first dimension represents the elements we have already used, the second represents the difference of the sets' sums and the third is rather technical and it used to be sure that we won't overwrite tuples that have wanted properties. In the cells we store the two sets of indices and the total sum of their weights. Moreover when the third dimension has the value 1 then this means that these sets could be a part of a feasible  solution.
\begin{algorithm}[H]
\caption{Case 2 solution [$\mathcal{SOL}_{Case 2}(A,m,S_{min},S_{max})$ function]}
\label{sub3}
\begin{algorithmic}[1]
  \Require a set $A=\{a_1,\ldots,a_{2\cdot n}\}$, $a_i\in \mathbb{Z}^+$ and an integer $m$, $1\leq m\leq 2\cdot n$ and $S_{min},S_{max} \subseteq \{1,...,n\} $.
  \Ensure Case 2 optimal solution for $Semi$-$Restricted$ $2 \mbox{-}Set \ SSR$.
        \State $S'_1\leftarrow  \emptyset$, $S'_2\leftarrow \emptyset$
        \If{$m\leq n$} 
            \State$p \leftarrow 0$, $p' \leftarrow n$
        \Else 
            \State $ p \leftarrow n$, $p' \leftarrow 0$ 
        \EndIf
        \State $Q\leftarrow a_m + \sum_{i \in S_{min}} a_{i+p} $
        \State $T[i,d,l] \leftarrow \{ \emptyset, \emptyset, 0 \}, \ \forall \ (i,d,l)\in \{0,\ldots, n\}\times \{-2\cdot Q,\ldots , Q\}\times \{0,1\} $
        \State $T[0,a_m,0]\leftarrow (\{m\},\emptyset,a_m)$
        \For{$i\leftarrow1$ \textbf{to} $n$}
            \ForAll{$(d,l) \in \{-2\cdot Q,\ldots , Q\}\times \{0,1\}$
    	        \State $(S_1,S_2,x)\leftarrow T[i-1,d,l] $}
    	        \State $T[i,d,l]\leftarrow \mathsf{LTST}(T[i,d,l], T[i-1,d,l])$ \label{line12}
                \State $d' \leftarrow d+a_{i+p}$
                \If{$i \in S_{min}$}
        	        \State $T[i,d',l]\leftarrow \mathsf{LTST}(T[i,d',l], (S_1\cup\{i+p\},S_2,x+a_{i+p}))$\label{line13}
                \EndIf
                \State $d' \leftarrow d-a_{i+p'}$
                \If{$i \in S_{max}$ \textbf{and} $d' \geq -2 \cdot Q$}
                    \State $T[i,d',1]\leftarrow \mathsf{LTST}(T[i,d',1], (S_1,S_2\cup\{i+p'\},x+a_{i+p'}))$ \label{line14}
                \ElsIf{$i\notin S_{max}$ \textbf{and} $d'\geq -2 \cdot Q$}
                    \State $T[i,d',l]\leftarrow \mathsf{LTST}(T[i,d',l], (S_1,S_2\cup\{i+p'\},x+a_{i+p'}))$ 
                \EndIf
            \EndFor
        \EndFor
        \For{$d \leftarrow -2 \cdot Q$ \textbf{to} $ Q$ }
      \State $(S_1,S_2,x)\leftarrow T[n,d,1] $
      \If{$\mathcal{MR}(S_1,S_2,A)< \mathcal{MR}(S'_1,S'_2,A)$}
	   \State $S'_1\leftarrow S_1$, $S'_2\leftarrow S_2$
      \EndIf
    \EndFor
    \State \Return $(S'_1,S'_2)$
\end{algorithmic}
\end{algorithm}  


\begin{theorem}
Algorithm~\ref{AlgTSm} runs in time $\mathcal{O}(n^2 \cdot a_m)$.
\end{theorem}

Since Algorithm~\ref{AlgTSm} is a pseudo - polynomial time algorithm for the $Semi$-$Restricted $ $2 \mbox{-} Set \ SSR$ which solves the instances with integer values and runs in time $\mathcal{O}(poly(n,a_m))$, by using Theorem~\ref{theor20} we get that $2 \mbox{-}Set \ SSR$ admits an $FPTAS$. Furthermore, by using Algorithm~\ref{Alg1} we have the following:

\begin{theorem}
\label{theor30}
For $2 \mbox{-}Set \ SSR$ and for every $\varepsilon \in (0,1)$ we can find an $(1+\varepsilon)$ approximation solution in time $\mathcal{O}(n^4/\varepsilon)$.
\end{theorem}

\end{section}

\begin{section}{Approximation of $SSR$ and $Factor\mbox{-}r \ SSR$}
\label{SecExamples}
In this section we will use the algorithm we design for the $2 \mbox{-}Set \ SSR$ in order to approximate the original problem $SSR$ and another one variation of $SSR$, the $Factor\mbox{-}r \ SSR$. 

Before we approximate these problems we will prove that both of them are in $F\mbox{-}SSR$. Starting with the $SSR$, it is easy to identify it with a problem $(\mathcal{I},\mathcal{F},\mathcal{M},\mathcal{G})$ in $F\mathcal{-}SSR$: we let the set of instances $\mathcal{I}$ contain sets of positive integers $A=\{a_1,\ldots,a_n\}$, the set of feasible solutions $\mathcal{F}$ contain all the pairs of sets $(S_1, S_2)$ such that 
$ S_1\cup S_2 \subseteq\{1,\ldots,n\}, \ S_1 \cap S_2 = \emptyset$,
the measure be $\mathcal{M}=\mathcal{MR}(S_1, S_2,A)$, and the goal function be $\mathcal{G}=\min$. 

Regarding $Factor\mbox{-}r \ SSR$, we identify it with a problem $(\mathcal{I},\mathcal{F},\mathcal{M},\mathcal{G})$ in $F\mathcal{-}SSR$, by letting the set of instances $\mathcal{I}$ contain sets of positive numbers $A=\{a_1,\ldots,a_{2n}\}$ \ $ = \{a_1,\ldots a_n, r \cdot a_1,\ldots ,r \cdot a_ n\}$ with $a_i \in \mathbb{Z}^+$ for $i \in \{1,\ldots,n\}$, $r \in \mathbb{R}$, the set of feasible solutions $\mathcal{F}$ contain all pairs of sets $(S_1, S_2)$ such that $ S_1 \subseteq \{1,...,n\} \mbox{ and } S_2 \subseteq \{n+1,...,2 n\} \nonumber \mbox{ and } \forall \ (i,j), i \in S_1 \wedge j\in S_2 \Rightarrow i+n \neq  j$, 
the measure be $\mathcal{M}=\mathcal{MR}(S_1, S_2,A))$, and the goal function be $\mathcal{G}=\min$.

For both problems we can modify their input in order to match the input of $2 \mbox{-}Set \ SSR$. Specifically, is not hard to see that an optimal solution for $SSR$ with input $A=\{a_1,\ldots ,a_n\}$ is an optimal solution for $2 \mbox{-}Set \ SSR$ with input $A=\{\{a_1,a_1)\ldots ,(a_n,a_n)\}$ and vice versa. The same applies to an optimal solution for $Factor\mbox{-}r \ SSR$ with input $(\{a_1,\ldots,a_n\},r)$ and an optimal solution of $2 \mbox{-}Set \ SSR$ with input $A=\{(a_1,r \cdot a_1)\ldots ,(a_n,r \cdot a_n)\}$. Furthermore the feasible solutions for $2 \mbox{-}Set \ SSR$, with the specific input we discussed above, are the same with the ones for $SSR$ (respectively for $Factor\mbox{-}r \ SSR$). So if we find an $(1+\varepsilon)$ approximating solution for the $2 \mbox{-}Set \ SSR$ problem with input $A=\{(a_1,a_1)\ldots ,(a_n,a_n)\}$ (resp. with input $A=\{(a_1,r \cdot a_1)\ldots ,(a_n,r \cdot a_n)\}$) then this is an $(1+\varepsilon)$ approximating solution for $SSR$ (resp. for $Factor\mbox{-}r \ SSR$).

\end{section}




%
%

\newpage

\begin{section}{APPENDIX}

\textbf{Lemma \ref{lemma1}.} \textit{Let $\mathcal{P}=(\mathcal{I},k,\mathcal{F}) $ a problem in $F\mbox{-}SSR$ and  $\mathcal{P'}=(\mathcal{I}',k',\mathcal{F}') $ the semi restricted version of $\mathcal{P}$. If $(E,w) \in \mathcal{I}$ and $(E,w',m) \in \mathcal{I}'$ are the instances of $\mathcal{P}$ and $\mathcal{P}'$ respectively then 
    any feasible solution $(S_1,\ldots, S_k)$  of the instance $(E, w',m)$ of $\mathcal{P}'$
    is also a feasible solution of the instance $(E,w)$  of $\mathcal{P}$.}
\begin{proof}
The feasible solutions $(S_1,\ldots, S_k)$ of instance $(E,w',m)$ of $\mathcal{P}'$ is the $\mathcal{F}'(k,E,w',m)$. By the definition of $\mathcal{P}'$ in $Semi$-$Restricted$ $F\mbox{-}SSR$ we have $\mathcal{F}'(k,E,w',m) \subseteq \mathcal{F}(k,E)$ thus the lemma holds.
\end{proof}

\noindent
\textbf{Lemma \ref{lemma2}.} \textit{ 
Let $\mathcal{P}=(\mathcal{I},k,\mathcal{F}) $ a problem in $F\mbox{-}SSR$ and  $\mathcal{P'}=(\mathcal{I}',k',\mathcal{F}') $ the semi restricted version of $\mathcal{P}$. If $E$ is a set of elements and $w$, $w'$ two weight functions such that:
$$\mbox{For any } i,j \in \{1,\ldots n\}, \ \  w(e_i)< w(e_j) \Leftrightarrow w'(e_i)\ \leq w'(e_j) $$
then any feasible solution $(S_1,\ldots, S_k)$ for the instance $(E,w)$  of $ \mathcal{P}$ is a feasible solution for the instance $(E,w',m)$  of $ \mathcal{P}'$ if 
$$w(e_m) = \min_{j \in \{1,\ldots,k\}}\{ \max\{ w(e_i) \mid i \in S_j\}\}$$}

\begin{proof}
Let $(S_1,\ldots, S_k)$ be a feasible solution for the problem $ \mathcal{P}$ with instance $(E,w)$. This mean that $(S_1,\ldots, S_k) \in \mathcal{F}(k,E)$. Assuming that: 
$$w(e_m) = \min_{j \in \{1,\ldots,k\}}\{ \max\{ w(e_i) \mid i \in S_j\}\}$$ 
is easy to see that $w(e_m) \leq \max\{ w(e_i) \mid i \in S_j\mbox{ and } j \in \{1,\ldots k\}\}$ which by the assumptions in the lemma gives us $w'(e_m) \leq \max\{ w'(e_i) \mid i \in S_j\mbox{ and } j \in \{1,\ldots k\}\}$. This means that $(S_1,\ldots, S_k)$ is a feasible solution for $ \mathcal{P}'$ with instance $(k,E,w',m)$ it meets both conditions $$(S_1,\ldots, S_k) \in \mathcal{F}(k,E)$$ and  
$$ w'(e_m)= \min_{j \in \{1,\ldots,k\}}\big\{ \max \{ w'(e_i) \mid i \in S_j \} \big\}. $$
\end{proof}

\noindent
The following three lemmas will be used to prove theorem~\ref{theor10}. We will start with the following lemma that relates $A$ with $A'$.

\begin{lemma}
\label{lemma1T10}
Let $a_i,a_i'$ and $\delta$ as they are in Theorem~\ref{theor10}, then for any $S \in \{S_{1Opt},$ $S_{2Opt},S_{1},S_{2} \}$ they satisfy the following:

\begin{align}
    \sum_{i \in S}a_i-m\cdot \delta &\leq \sum_{i \in S}a'_i\cdot \delta \leq \sum_{i \in S}a_i \label{ineqai} \\
     m \cdot \delta &\leq
    \frac{\varepsilon}{3} 
    \sum_{i \in S}{a_i}\label{ineqmd}
\end{align}
\end{lemma}

\begin{proof}
To prove Eq.~\ref{ineqai}, notice that for all $i \in \{1,\ldots,n\}$ we define  $a'_i=\lfloor \frac{a_i}{\delta} \rfloor$. This gives us
$$
\frac{a_i}{\delta}-1 \leq a'_i \leq \frac{a_i}{\delta} \Rightarrow a_i -\delta \leq \delta \cdot a'_i \leq a_i \ . $$
In addition for any $S \in \{S_{1Opt},S_{2Opt},S_{1},S_{2} \}$ we have $|S| \leq m $, which means that
$$\sum_{i \in S}a_i-m\cdot \delta\leq \sum_{i \in S}a_i-|S| \cdot \delta \leq \sum_{i \in S}a'_i\cdot \delta \leq \sum_{i \in S}a_i \ .$$
As for Eq.~\ref{ineqmd} we have to take into account the theorem's assumptions. Specifically we know that
$$ m \leq \sum_{i \in S} a_i \mbox{ for all } S \in \{S_{1Opt},S_{2Opt},S_{1},S_{2}\} $$
which gives 
$$
    m \cdot \delta   =    \frac{\varepsilon \cdot w }{3}    \leq    \frac{\varepsilon}{3}     \sum_{i \in S}{a_i}
$$
\end{proof}

\begin{lemma}
\label{lemma2T10}
For sets $S_1$ and $S_2$ it holds
\begin{align*}
    \mathcal{MR}(S_1,S_2,A) \leq \mathcal{MR}(S_1,S_2,A') + \frac{\varepsilon}{3}
\end{align*}
\end{lemma}

\begin{proof} 
Without loss of generality we will assume that 
\begin{align*}
    \mathcal{MR}(S_{1},S_{2},A) = \mathcal{R}(S_{1},S_{2},A)
\end{align*}
\begin{align*}
    \mathcal{R}(S_1,S_2,A)= \frac{\sum_{i \in S_1}a_i}{\sum_{j \in S_2}a_j} &\leq  \frac{\sum_{i \in S_1}a'_i\cdot \delta + \delta \cdot m}{\sum_{j \in S_2}a_j} &\mbox{[by Eq.~\ref{ineqai}] }\\
    &\leq \frac{\sum_{i \in S_1}a'_i\cdot \delta}{\sum_{j \in S_2}a'_j} +\frac{ \delta \cdot m}{\sum_{j \in S_2}a_j} &\mbox{[by Eq.~\ref{ineqai}] }\\
    &\leq \mathcal{MR}(S_1,S_2,A') +\frac{\varepsilon}{3} &\mbox{[by Eq.~\ref{ineqmd}] }
\end{align*}
\end{proof}

\begin{lemma}
\label{lemma3T10}
For every $\varepsilon \in (0,1)$ we have that 
\begin{align*}
    \mathcal{MR}(S_{1Opt},S_{2Opt},A') \leq (1+\varepsilon/2) \cdot \mathcal{MR}(S_{1Opt},S_{2Opt},A)
\end{align*}
\end{lemma}

\begin{proof}
Without loss of generality we will assume that 
\begin{align*}
    \mathcal{MR}(S_{1Opt},S_{2Opt},A') = \mathcal{R}(S_{1Opt},S_{2Opt},A')
\end{align*}
From Eq.~\ref{ineqai} we have that
\begin{align*}
 \mathcal{MR}(S_{1Opt},S_{2Opt},A') &= \frac{\sum_{i \in S_{1Opt}}{a'_i}}{\sum_{i \in S_{2Opt}}{a'_i}}
    \leq 
    \frac{\sum_{i \in S_{1Opt}}{a_i}}{\sum_{i \in S_{2Opt}}{a_i}-m \cdot \delta}
    \\
    &=
    \frac{\sum_{i \in S_{1Opt}}{a_i}}{\sum_{i \in S_{2Opt}}{a_i}-m \cdot \delta}
    \cdot
     \frac{\sum_{i \in S_{2Opt}}{a_i}}{\sum_{i \in S_{2Opt}}{a_i}}
     \\
     &=
    \frac{\sum_{i \in S_{2Opt}}{a_i}}{\sum_{i \in S_{2Opt}}{a_i}-m \cdot \delta}
    \cdot  \frac{\sum_{i \in S_{1Opt}}{a_i}}{\sum_{i \in S_{2Opt}}{a_i}}
     \\
     & = \left(1+ \frac{m\cdot \delta}{\sum_{i \in S_{2Opt}}a_i - m\cdot \delta} \right) \cdot \mathcal{R}(S_{1Opt},S_{2Opt},A)\
\end{align*}
\noindent
by Eq.~\ref{ineqmd} it follows that 
\begin{align*}
 \mathcal{MR}(S_{1Opt},S_{2Opt},A') 
 &\leq  \left(1 + \frac{1}{ \frac{3}{\varepsilon }- 1}\right) \cdot \mathcal{R}(S_{1Opt},S_{2Opt},A)& \\
 &=
 \left(1 + \frac{\varepsilon}{ 3- \varepsilon}\right)\cdot \mathcal{R}(S_{1Opt},S_{2Opt},A)
 \\
 & \leq
 \left(1 + \frac{\varepsilon}{ 2}\right)\cdot \mathcal{R}(S_{1Opt},S_{2Opt},A)&  \mbox{[because }\varepsilon \in (0,1)\mbox{]}
 \\
 & \leq 
 \left(1 + \frac{\varepsilon}{ 2}\right)\cdot \mathcal{MR}(S_{1Opt},S_{2Opt},A).&
\end{align*}
This concludes the proof.

\end{proof}

\noindent
Now we are ready to prove Theorem~\ref{theor10}.
\begin{proof}(of Theorem 1)
The theorem follows from a sequence of inequalities:
\begin{align*}
 \mathcal{MR}(S_{1},S_{2},A) 
 &\leq \mathcal{MR}(S_{1},S_{2},A')+ \frac{\varepsilon}{3}& \mbox{[by Lemma }\ref{lemma2T10}\mbox{]} \\
 & \leq \mathcal{MR}(S_{1Opt},S_{2Opt},A')+ \frac{\varepsilon}{3}&\\
 & \leq (1+\frac{\varepsilon}{2})\cdot \mathcal{MR}(S_{1Opt},S_{2Opt},A) + \frac{\varepsilon}{3}& \mbox{[by Lemma }\ref{lemma3T10}\mbox{]}\\
 & \leq (1+\varepsilon)\cdot \mathcal{MR}(S_{1Opt},S_{2Opt},A).&
\end{align*}
\end{proof}

\noindent
\textbf{Theorem~\ref{theor20}.} \textit{Let $\mathcal{P} = (\mathcal{I},\mathcal{F},\mathcal{M},\mathcal{G}) $ be a problem in $F$-$SSR$ and   $\mathcal{P}'  = (\mathcal{I}',\mathcal{F}',$ $\mathcal{M},\mathcal{G})$ its corresponding problem in $Semi \mbox{-} Restricted \ F\mbox{-}SSR$.
If for problem $\mathcal{P}'$ there exists an algorithm that solves exactly all instances $A = \{a_1,\ldots a_n,m\} \in \mathcal{I}'$ in which all $a_i$ values are integers in time $\mathcal{O}(poly(n,a_m))$, then $\mathcal{P}$ admits an FPTAS.}

\begin{proof}
Let us first remind that for every  problem in $F\mbox{-}SSR$, two different instances $A$ and $A'$ with the same number of elements have exactly the same feasible solutions (see Remark~\ref{remIndex}).
We need to prove that the output of  Algorithm~\ref{Alg1} is a $(1+\varepsilon)$-approximation of the optimum solution of $\mathcal{P}$ with input $A= \{a_1,\ldots a_n\}$.
Let $\varepsilon \in (0,1)$, $S_1,...,S_k$ be the sets of the optimum solution of $\mathcal{P}$ and $S_1^{(m)},...,S_k^{(m)}$ the solution of $\mathcal{P}'$ with input $(A^{(m)},m) = (\{a_1,\ldots a_n\},m) \in \mathcal{I}'$. Here we have to remind that a feasible solution of $\mathcal{P}'$ is a feasible solution of $\mathcal{P}$ if the input sets $A$' and $A$ for the problems have the same size. So the optimal solution $S_1^{(m)},\ldots,S_k^{(m)}$ of the  $\mathcal{P}'$ with input $(A^{(m)},m) = (\{a_1,\ldots a_n\},m) \in \mathcal{I}'$ is a feasible solution of $\mathcal{P}$ with input $A= \{a_1,\ldots a_n\}$. We will also denote with $a_{n_0} \in A$ the minimum element among the maximum of the sets $S_1,\ldots,S_k$ of the optimal solution, i.e. 
\begin{align*}
     a_{n_0} = \min_{j \in \{1,\ldots, k\}}{\left(\max_{i \in S_j}{a_i} \right) }
\end{align*}
This means that for the output of the Algorithm~\ref{Alg1}, $S_1^*,\ldots,S_k^*$ we have  $$\mathcal{MR}(S_1^*,...,S_k^*,A) \leq \mathcal{MR}(S_1^{(n_0)},...,S_k^{(n_0)},A) $$ 
so it is sufficient to prove that
$$\mathcal{MR}(S_1^{(n_0)},...,S_k^{(n_0)},A) \leq (1+ \varepsilon)\cdot \mathcal{MR}(S_1,...,S_k,A)
$$ 
It is obvious that the optimal solution $S_1,...,S_k$ is a feasible solution for the problem $\mathcal{P}'$ with input $A^{(n_0)}$ so we have
\begin{align}
\label{eqProofT2}
    \mathcal{MR}(S_1^{(n_0)},...,S_k^{(n_0)},A^{(n_0)}) \leq \mathcal{MR}(S_1,...,S_k,A^{(n_0)})
\end{align}
Without loss of generality let $\mathcal{MR}(S_1^{(n_0)},S_2^{(n_0)},A)=\mathcal{MR}(S_1^{(n_0)},\ldots, S_k^{(n_0)},A)$ and $\mathcal{MR}(S_1,S_2,A^{(n_0)})=\mathcal{MR}(S_1,\ldots, S_k,A^{(n_0)})$.
Then due to the definition of the $\mathcal{MR}$ function

\begin{align*}
    \mathcal{MR}(S_1^{(n_0)},S_2^{(n_0)},A^{{(n_0)}}) & \leq \mathcal{MR}(S_1^{(n_0)},\ldots, S_k^{(n_0)},A^{{(n_0)}}) & \\
    & \leq \mathcal{MR}(S_1,\ldots, S_k,A^{{(n_0)}}) & \mbox{[by Eq.~\ref{eqProofT2}]} \\
    & = \mathcal{MR}(S_1,S_2,A^{{(n_0)}}) &
\end{align*}
From the above equation and the definition of $A^{(n_0)}$ it is easy to see that the pairs of sets $(S_1^{(n_0)},S_2^{(n_0)})$ and $(S_1,S_2)$ satisfy the requirements of Theorem~\ref{theor10} which gives us that 

\begin{align*}
    \mathcal{MR}(S_1^{(n_0)},\ldots, S_k^{(n_0)},A) & = \mathcal{MR}(S_1^{(n_0)},S_2^{(n_0)},A) & \mbox{[by assumption]} \\
    & \leq (1+\varepsilon)\cdot \mathcal{MR}(S_1, S_2,A) & \mbox{[by Theorem~\ref{theor10}]}\\
    & \leq (1+\varepsilon) \cdot \mathcal{MR}(S_1,\ldots, S_k,A) \ .
\end{align*}

As for the running time we have that this algorithm begins with a for loop that runs $n$ times. In each iteration the algorithm computes $A^{(m)}$ in time $\mathcal{O}(n)$, runs the algorithm $\mathcal{SOL}_{ex,\mathcal{P}_m}$ in $poly(n,a_m')$ and finally it has to evaluate $\mathcal{MR}(S_1',\ldots,$ $S_k',A)$ and $\mathcal{MR}(S_1^*,\ldots,S_k^*,A)$. This evaluation takes time $\mathcal{O}(k^2)=\mathcal{O}(n^2)$ due to $k \leq n$ (because $S_1,\ldots,S_k$ are disjoint). So one iteration takes time $poly(n,a'_m)$. Thereby Algorithm~\ref{Alg1} runs in $poly(n,a'_m)$.  

We will prove that the $a'_m$ which we use in each iteration is polynomially bounded by $n$ and $1/\varepsilon$.
We have that $a'_m= \lfloor a_m/\delta \rfloor= \lfloor 3\cdot n \cdot a_m/\varepsilon \cdot a_m \rfloor \leq 3\cdot n /\varepsilon$. Hence the running time is $poly(n, 1/\varepsilon)$ proving that  Algorithm~\ref{Alg1} is an FPTAS for  problem $\mathcal{P}$.

\end{proof}

\noindent
\textbf{Proof of theorem~\ref{theor30}.} \textit{Algorithm~\ref{AlgTSm} runs in time $\mathcal{O}(n^2 \cdot a_m)$.}
\begin{proof}
Observe that in Algorithm~\ref{AlgTSm} we initialize our variables and we select $S_{min}$ and $S_{max}$ according to m. These selection take time $\mathcal{O}(n)$. Then we run two algorithms (Algorithm~\ref{sub2} and Algorithm~\ref{sub3}). Specifically Algorithm~\ref{sub2} runs in $\mathcal{O}(n)$ due to the fact that the cardinality of $S_{max}$ can not be greater than $n$. Furthermore in Algorithm~\ref{sub3} we fill a matrix with size $n \times 3\cdot Q \times 2$ and by using suitable data structure, we can store the sets in time (and space) $\mathcal{O}(1)$ per cell. This implies that Algorithm~\ref{sub3} runs in $\mathcal{O}(n\cdot Q)$. Last, is easy to see that $Q =a_m + \sum_{i \in S_{min}} a_{i+p}$ and $a_i\leq a_m$ for every $a_i$ this sum which gives us that $Q \leq n \cdot a_m$. So the Algorithm~\ref{AlgTSm} runs in time $\mathcal{O}(n^2 \cdot a_m)$.
\end{proof}

\begin{theorem}
The Algorithm~\ref{AlgTSm} returns an optimal solution for the semi restricted version of $2 \mbox{-}Set \ SSR$.
\end{theorem}

\begin{proof}

Before we start the proof we have to remark two things. The first is that generally, it's not necessary that $S_1 \subseteq  \{1,\ldots ,n\}$ and $S_2 \subseteq \{n+1,\ldots ,2 \cdot n\}$ but they may be reversed, so when we mention a feasible solution, by convention we will regard that $m \in S_1$. Secondly, in the case-algorithms inside of Algorithm~\ref{AlgTSm} we construct two variables $p$ and $p'$ such that:
\begin{align*}
& \bullet \mbox{ If }  m \in \{1, \ldots ,n\} \mbox{ then }  (p,p')=(0,n) \\
& \bullet \mbox{ If } m \in \{n+1, \ldots ,2\cdot n\} \mbox{ then } (p,p')=(n,0)
\end{align*}
With the use of these variables we will prove some properties for the indices of the feasible solutions.

\begin{lemma}
Let $(S_1,S_2)$ be a feasible solution of the problem with input $(\{a_1,\ldots,$ $ a_{2 \cdot n}\},m)$, the set $S = \{1,\ldots ,n\}$ and $(p,p')$ the variables we defined above, then:
\begin{align}
    & \bullet \ S_1 \subseteq \{i+p\mid i \in S\}
     \label{s1indexs} \\
    & \bullet \ S_2 \subseteq \{i+p'\mid i \in S\}
     \label{s2indexs} \\
    & \bullet \mbox{ for an index } j\in S_1 \mbox{ then } j-p+p' \notin S_2  \label{relation-s1s2} \\
    & \bullet \mbox{ for an index } j\in S_2 \mbox{ then } j-p'+p \notin S_1 \label{relation-s2s1}
\end{align}

\end{lemma}

\begin{proof}
Because the feasible solutions $(S_1,S_2)$ are in such order so $m \in S_1$ and if we keep in mind that one of these sets is subset of $\{1,\ldots ,n\}$ and the other is subset of $\{n+1,\ldots ,2 \cdot n\}$ we have that\\
$\bullet$ If $m \in \{1, \ldots ,n\}$ then $S_1 \subseteq \{1,\ldots ,n\}$, $S_2 \subseteq \{n+1,\ldots ,2 \cdot n\}$ and $(p,p')=(0,n)$\\
$\bullet$ If $m \in \{n+1, \ldots ,2\cdot n\}$ then $S_2 \subseteq \{1,\ldots ,n\}$, $S_1 \subseteq \{n+1,\ldots ,2 \cdot n\}$ and $(p,p')=(n,0)$ \\ 
Without loss of generality we assume that $m \leq n$, this means that $p=0$,
$$S_1 \subseteq \{1,\ldots ,n\} = \{i +p \mid i \in S\}$$
and because $p'=n$ 
$$ S_2 \subseteq \{n+1,\ldots ,2 \cdot n\} = \{i +p' \mid i \in S\} \ .$$
It remains to prove the relations between the indices of the two sets. By the definition of the problem for any feasible solution $(S_1,S_2)  \nexists\ (i,j)$ such that $i \in  S_1$, $j \in  S_2$ and $i \equiv j\  ( mod \ n) $. Because $ p,p' \in \{0,n\}$ we have that 
$ j \equiv j-p+p' \equiv j-p'+p \ (mod \ n) $ so both of the last two properties holds.
\end{proof}

\begin{lemma}
Let $(S_1,S_2)$ be a feasible solution for the problem with input $(\{a_1,\ldots $ $,  a_{2 \cdot n}\},m)$ and $S_{min}$, $S_{max}$ the sets as they are defined in Algorithm~\ref{AlgTSm}, then the following are true:
\begin{align}
    &  \exists\  i \in S_{max} \mbox{ such that } i+p' \in S_2 \label{s2smax} \\
    & S_1 \subseteq \{i+p \mid i \in S_{min}\} \cup \{m\}  \label{s1smin}
\end{align}
\end{lemma}

\begin{proof}
We will start with the definitions of $S_{min}$ and $S_{max}$.
\begin{align*}
    & S_{min} =\{ i \mid  i \in \{1,\ldots ,n \} \mbox{ and } a_{i+p} \leq a_m\} \smallsetminus \{m - p\} \\
    & S_{max} =\{ i \mid  i \in \{1,\ldots ,n \} \mbox{ and } a_{i+p'} > a_m\} \smallsetminus \{m -p + p'\} \ .
\end{align*}
For any feasible solutions of the semi restricted version of $2 \mbox{-}Set \ SSR$ we know that $a_m = \max_{i \in S_1} a_i$ and $a_m \leq \max_{i \in S_2} a_i $. 
Let $ j \in S_{2} \mbox{ such that } a_j \geq a_m $ then by relation~\ref{s2indexs}, we have that
 $\exists\  i \in \{1,\ldots ,n \} \mbox{ such that } i+p' =j \mbox{, } a_{i + p'} \geq a_m \mbox{ and } i \neq m -p+p' $ so $i \in S_{max}$ (by its definition). So relation~\ref{s2smax} holds.
Now by considering the relation~\ref{s1indexs}, if $j \in S_{1}\smallsetminus\{m\}$ then $ \exists \ i \in \{1,\ldots ,n \}$ such that $ i+p =j \mbox{, } a_{i + p} \leq a_m \mbox{ and } j \neq m $ which means that $i \in S_{min} $
so the relation~\ref{s1smin} holds. Thus the lemma holds.
\end{proof}

Now, let $(S^*_1, S^*_2)$ be an optimal solution for the semi restricted version of $2 \mbox{-}Set  \ SSR$ with input $A=\{a_1,\ldots a_{2 \cdot n} \}$ and $m$. Without lost of generality let 
$$\max_{i \in S^*_1}\{a_i\}=a_{m_1}=a_m< a_{m_2}=\max_{j \in S^*_2}\{a_j\}$$
this means that the sets appear in the same order as if they were constructed from the algorithm.
For this optimal solution we have two cases, either $a_{m_2} > Q$ or $a_{m_2} \leq Q$ (where $Q \leftarrow a_m + \sum_{i \in S_{min}} a_{i+p}$ as it is defined in case-algorithms of Algorithm~\ref{AlgTSm}).

$\bullet$ Case 1 ($a_{m_2} > Q$): In this case we will return a solution with ratio equal to the optimal using Algorithm~\ref{sub2}. By relation~\ref{s2indexs} we know that there exists $ m_0 \in \{1,\ldots,n\}$ such that the index $m_0+p'=m_2$. With the additional knowledge that $a_{m_2} > Q$ we have that $m_0 \in S_{max}$. We claim that the pair of sets $(S,\{m_2\})$, where 
$S=\{i+p \mid i \in S_{min} \mbox{ or } i=m-p \} \smallsetminus \{m_{0}+p \}$ is an optimal solution for this case. 
In order to prove this claim we need to observe that for any feasible solution $(S_1,S_2)$ for this case we have
$S_1 \subseteq (\{i+p \mid i \in S_{min}\} \cup \{m\}) \smallsetminus \{ m_2 - p' +p\} $ and $S_2 \supseteq \{m_2\}$.  So for any feasible solution for this case we have:
\begin{align*}
\mathcal{MR}(S_1,S_2,A) & = \mathcal{R}(S_2,S_1,A) \geq \mathcal{R}(\{m_2\},S_1,A) \\ & \geq 
\mathcal{R}(\{m_2\},(\{i+p \mid i \in S_{min}\} \cup \{m\}) \smallsetminus \{ m_2 - p'+p\} ,A)
\end{align*}
which proves the claim. Due to this fact, Algorithm~\ref{sub2} returns $(S,\{m_2\})$ or a pair of sets with the same max ratio.

$\bullet$  Case 2 ($a_{m_2} < Q$): The first thing we have to prove in this case is the following lemma,
\begin{lemma} \label{lemmadif}
If $(S'_1,S'_2)$ is a feasible solution for this case and $(S_1,S_2)$ is a pair of sets such that $S_1 \subseteq S'_1$ and $S_2 \subseteq S'_2$ then:
\begin{align*}
    -2\cdot Q \leq \sum_{i \in S_1} a_i - \sum_{j \in S_2} a_j \leq Q
\end{align*}
\end{lemma}
\begin{proof}
By relation~\ref{s1smin} it is obvious that $\sum_{i \in S_1} a_i \leq Q$ so we need prove that $$\sum_{j \in S_2} a_j \leq 2\cdot Q \ .$$
Let's assume that $\sum_{j \in S_2} a_j > 2\cdot Q$ then because $ Q \geq a_{m_2} = \max_{i \in S'_2}\{a_i\}$ and $S_2 \subseteq S'_2$, we have that $S_2$ should contain at least 3 indices and the same holds for the set $S'_2$. Let $m_0 \neq m_2$ be one of them, then because $Q \geq a_{m_0} $ and $\sum_{j \in S'_2} a_j \geq \sum_{j \in S_2} a_j > 2\cdot Q$ we have that
$$\sum_{j \in S'_2} a_j > \sum_{j \in S'_2 \smallsetminus \{m_0\}} a_j> Q $$
and because $\sum_{i \in S'_1 } a_i \leq Q $  
\begin{align*}
    1 \leq \mathcal{MR}(S'_1,S'_2 \smallsetminus \{m_0\},A) \leq  \mathcal{MR}(S'_1,S'_2,A)
\end{align*}
which is a contradiction.
\end{proof}

With the next lemma we will prove that for any feasible solution $(S_1,S_2)$ the cell $T[n,d,1]$
where $d= \sum_{i \in S_1}a_i -\sum_{i \in S_2}a_i$ is non empty. Furthermore the sets which are stored in this cell have max ratio at most $\mathcal{MR}(S_1,S_2,A)$. 
\begin{lemma} \label{lemma_not_empty}
Let  $(S^*_1,S^*_2)$ is a feasible solution for this case and $(S_1,S_2)$ is a pair of sets such that $S_1 \subseteq S^*_1$, $S_2 \subseteq S^*_2$  and $m \in S_1$. We define $m_0 = \max \{0, i \mid i+p \in S_1 \smallsetminus \{m\} \mbox{ or } i+p' \in S_2\} $ and $d= \sum_{i \in S_1}a_i -\sum_{i \in S_2}a_i$.
\begin{align*}
    &\mbox{ If } S_2\cap \{i+p' \mid i \in S_{max} \} =\emptyset &\mbox{ then the cell } T[m_0,d,0] \neq (\emptyset,\emptyset,0)\\
    &\mbox{ else } &\mbox{ then the cell } T[m_0,d,1] \neq (\emptyset,\emptyset,0)
\end{align*}
\end{lemma}

\begin{proof}
For any pair $(S_1,S_2)$ (as these described at the lemma) we define the following set
 $$ S_{S_1,S_2}=\{0\} \cup \{i \ mod \ n  \mid i+p \in S_1 \mbox{ or } i+p' \in S_2)  \} \smallsetminus\{m -p\} \ . $$
We have to notice that $i \in \{1,\ldots,n\}$ if $i+p \in S_1$ (by relation~\ref{s1indexs}) and the same holds if $i+p' \in S_2$. Because of this we know that $S_{S_1,S_2} \subseteq \{0,\ldots ,n\}\smallsetminus\{m -p\}$.
Now, we will prove this lemma by using strong induction to the maximum element of the set $S$.\\ 
$\bullet$ If $\max\{S_{S_1,S_2}\}=0$ (base case)\\
Because we have requested $m\in S_1$ and by the fact that $\max\{S_{S_1,S_2}\}=0$ we can conclude that  $S_2 \cap \{i+p' \mid i \in S_{max} \}= \emptyset$ and $(S_1,S_2) = (\{m\},\emptyset)$, which is the pair of sets the algorithm stores in the cell $T[0, a_m,0]$. This concludes the base case.\\
$\bullet $ Assuming that lemma's statement holds for all the indices $k'$ which are smaller or equal than a specific index $k<n$, we will now prove it for $k+1$.\\
$\bullet$ Let $k+1 = \max\{S_{S_1,S_2}\}$ this means that  $k+1+p \neq m$ and either $k+1+p \in S_1$ or $k+1+p' \in S_2$. So we have to check both cases.\\
\textbf{Case A $(k+1+p \in S_1)$.} In this case for the pair of sets $(S_1',S_2) = (S_1 \smallsetminus \{k+1+p\}, S_2)$ meets the conditions of the induction because  $\max\{S_{S_1',S_2}\} <\max\{S_{S_1,S_2} \smallsetminus \{k+1\}\} $. So for $d= \sum_{i \in S_1'}a_i - \sum_{i \in S_2}a_i $  we know that:
\begin{align*}
    &\mbox{ either the cell} & T[\max\{S_{S_1',S_2}\},d,0 ]\neq (\emptyset,\emptyset,0) \mbox{ if } S_2\cap \{i+p' \mid i \in S_{max} \}=\emptyset & \\
    &\mbox{ or the cell} & T[\max\{S_{S_1',S_2}\},d,1 ]\neq (\emptyset,\emptyset,0) \mbox{ if } S_2\cap \{i+p' \mid i \in S_{max} \} \neq \emptyset & \\
    & & \mbox{ (respectively)}
\end{align*}
Algorithm~\ref{sub3} moves up all the cells (line 13). This means that the cell $T[k,d,0 ]$ (resp. $ T[k,d,1]$) is non equal to $(\emptyset,\emptyset,0)$ in the case $S_2\cap S_{max}=\emptyset$ (resp. $S_2\cap S_{max} \neq \emptyset$). Because $k+1 \in S_{min}$ (by relation~\ref{s1smin} and $k+1+p \in S_1$) then Algorithm~\ref{sub3} fills the cell $T[k+1,d+a_{k+1+p},0]$ (resp. $T[k+1,d+a_{k+1+p},1]$) in line 16. This proves the case A\\
\textbf{Case B $(k+1+p' \in S_2)$.} Here we have two extra cases, either $k+1 \in S_{max}$ or not.\\
\textbf{Case B.1 $(k+1 \notin S_{max})$.}
Like in the previous case, the pair of sets $(S_1,S_2') = (S_1 , S_2\smallsetminus \{k+1+p'\})$ meets the conditions of the induction because  $\max\{S_{S_1,S_2'}\} <\max\{S_{S_1,S_2} \smallsetminus \{k+1\}\} $. So for $d= \sum_{i \in S_1}a_i - \sum_{i \in S_2'}a_i $ we know that:
\begin{align*}
    &\mbox{ either the cell} & T[\max\{S_{S_1,S_2'}\},d,0 ]\neq (\emptyset,\emptyset,0) \mbox{ if } S_2'\cap \{i+p' \mid i \in S_{max} \}=\emptyset & \\
    &\mbox{ or the cell} & T[\max\{S_{S_1,S_2'}\},d,1 ]\neq (\emptyset,\emptyset,0) \mbox{ if } S_2'\cap \{i+p' \mid i \in S_{max} \} \neq \emptyset & \\ 
    & & \mbox{ (respectively)}
\end{align*}
As we have said before, Algorithm~\ref{sub3} moves up all the cells (line 13). This means that the cell $T[k,d,0 ]$ (resp. $ T[k,d,1]$) is non equal to $(\emptyset,\emptyset,0)$ in the case of $S_2'\cap \{i+p' \mid i \in S_{max} \}=\emptyset$ (resp. $S_2'\cap \{i+p' \mid i \in S_{max} \} \neq \emptyset$) and because we know that $d -a_{k+1+p'} \geq - 2\cdot Q$ (by lemma~\ref{lemmadif}) in the line 22 the algorithm fills the wanted cell $T[k,d -a_{k+1+p'},0 ]$ (resp. $T[k,d -a_{k+1+p'},1 ]$).\\
\textbf{Case B.2 $(k+1 \in S_{max})$.} This case is similar to Case B.1 with the exception that the algorithm ensures that the algorithm fills the cell
$T[k,d -a_{k+1+p'},1 ]$ in line 20 if either of the $T[k-1,d,0 ]$ or $T[k-1,d,1 ]$ are non-empty.
\end{proof}

To complete the proof we have to observe three things.\\
First, in any cell Algorithm~\ref{sub3} keeps the pair of sets with the greater total sum.\\ Second, For all the pairs of sets $(S_1,S_2)$ which are stored we have $S_1 \subseteq \{i+p \mid i \in S_{min}\}$, $S_2 \subseteq \{i+p' \mid i \in \{1,\ldots,n\}\}$ and $\nexists (j,j')$ such that $j \in S_1$, $j' \in S_2$ and $j \equiv j' \ (mod \ n)$ (because we use only $i+p$ in $S_1$ or $i+p'$ in $S_2$ every time). 
Third, in order to store a pair of sets in any cell $T[i,d,1]$ we have either add a $i+p'$ with $i \in S_{max}$ to the $S_2$ or use an other cell $T[i',d',1]$ (which already have such an index in $S_2$) so these pairs are feasible solutions for this case.\\
With all that in mind we know that, if $(S^*_1,S^*_2)$ be an optimal solution for the problem then for $d = \sum_{i \in S^*_1}a_i - \sum_{i \in S^*_2}a_i $ ($d \in \{-2\cdot Q ,\ldots ,Q\}$ by the lemma~\ref{lemmadif}) a cell $T[i,d,1] \neq (\emptyset,\emptyset,0)$ due to lemma~\ref{lemma_not_empty} and relation~\ref{s2smax} so the same holds for the $T[n,d,1]$ (because the algorithm moves up all the cells). Let $(S_1,S_2)$ pair of sets stored in that specific cell. As we mentioned earlier the pair $(S_1,S_2)$ is a feasible solution such that:
$$ \sum_{i \in S_1}a_i -\sum_{i \in S_2}a_i =  \sum_{i \in S^*_1}a_i -\sum_{i \in S^*_2}a_i $$
and
$$\sum_{i \in S_1}a_i +\sum_{i \in S_2}a_i \geq  \sum_{i \in S^*_1}a_i +\sum_{i \in S^*_2}a_i $$
Now, because the differences of the sums of the pairs are the same, it is easy to see that the pair with the smaller max ratio is the one with the greater total sum. So, because we can not have smaller max ratio than the optimal this means that the stored pair is an optimal one. Thus the Algorithm~\ref{AlgTSm} returns the optimal solution in both of cases.
\end{proof}

\end{section}

\end{document}